\documentclass[twocolumn,journal]{IEEEtran}

\usepackage{graphicx}
\usepackage{latexsym}
\usepackage{amsmath}
\usepackage{epsfig, relsize}

\newtheorem{lemma}{Lemma}

\newcommand{\x}{\mbox{\boldmath $x$}}
\newcommand{\n}{\mbox{\boldmath $n$}}

\newcommand{\she}{\mbox{\boldmath ${h}$}}

\newcommand{\vzero}{\mbox{\boldmath $0$}}
\newcommand{\y}{\mbox{\boldmath $y$}}
\newcommand{\h}{\mbox{\boldmath $h$}}

\newcommand{\calH}{\mbox{\boldmath ${\mathcal H}$}}
\newcommand{\calX}{\mbox{\boldmath ${\mathcal X}$}}

\newcommand{\calY}{\mbox{\boldmath ${\mathcal Y}$}}
\newcommand{\calN}{\mbox{\boldmath ${\mathcal N}$}}

\newcommand{\beqn}{\begin{eqnarray*}}
\newcommand{\eeqn}{\end{eqnarray*}}

\newcommand{\II}{\mbox{\boldmath $I$}}

\newcommand{\QQ}{\mbox{\boldmath $Q$}}

\newcommand{\BB}{\mbox{\boldmath $B$}}

\newcommand{\BSigma}{\mbox{\boldmath $\Sigma$}}
\newcommand{\BD}{\mbox{\boldmath $D$}}

\newcommand{\A}{\mbox{\boldmath $A$}}
\newcommand{\BA}{\mbox{\boldmath $A$}}
\newcommand{\BF}{\mbox{\boldmath $F$}}
\newcommand{\BG}{\mbox{\boldmath $G$}}

\newcommand{\BP}{\mbox{\boldmath $P$}}

\newcommand{\BR}{\mbox{\boldmath $R$}}

\newcommand{\Ba}{\mbox{\boldmath $a$}}

\newcommand{\BX}{\mbox{\boldmath $X$}}

\newcommand{\BE}{\mbox{\boldmath $E$}}
\newcommand{\BW}{\mbox{\boldmath $W$}}

\newcommand{\BI}{\mbox{\boldmath $I$}}

\newcommand{\BU}{\mbox{\boldmath $U$}}

\newcommand{\Bg}{\mbox{\boldmath $g$}}

\newcommand{\bq}{\begin{eqnarray}}
\newcommand{\eq}{\end{eqnarray}}
\newcommand{\ba}{\begin{array}}
\newcommand{\ea}{\end{array}}

\newcommand{\La}{\mbox{\boldmath $\Lambda$}}

\begin{document}

\title{LOW COMPLEXITY BLIND EQUALIZATION FOR OFDM SYSTEMS WITH GENERAL CONSTELLATIONS}

\author{\emph{Tareq~Y.~Al-Naffouri$^{*1}$\footnote{$^*$corresponding author. Address: KFUPM, P.O.Box 1083, Dhahran 31261, Saudi Arabia. email: naffouri@kfupm.edu.sa,  tel. off.: +966 03 860 1030, fax: +966 03 860 3535. }, Ala Dahman$^1$, Muhammad S. Sohail$^1$, Weiyu Xu$^2$ and Babak Hassibi$^3$}\\
\large $^1$ Department of Electrical Engineering, King Fahd University of Petroleum and Minerals, Saudi Arabia\\ $^2$  Department of Electrical and Computer Engineering, University of Iowa, Iowa City, IA \\ $^3$ Department of Electrical Engineering, California Institute of Technology, Pasadena, CA}


\date{ }
\maketitle
\begin{abstract}
This paper proposes a low-complexity algorithm for blind equalization of data in OFDM-based wireless systems with general constellation. The proposed algorithm is able to recover data even when the channel changes on a symbol-by-symbol basis, making it suitable for fast fading channels. The proposed algorithm does not require any statistical information of the channel and thus does not suffer from latency normally associated with blind methods. We also demonstrate how to reduce the complexity of the algorithm, which becomes especially low at high SNR. Specifically, we show that in the high SNR regime, the number of operations is of the order $O(LN)$, where $L$ is the cyclic prefix length and $N$ is the total number of subcarriers. Simulation results confirm the favorable performance of our algorithm.
\end{abstract}



\section{Introduction} \label{sec:intro}

Modern wireless communication systems are expected to meet an ever increasing demand for high data rates. A major hindrance
for such high data rate systems is multipath fading. Orthogonal frequency division multiplexing (OFDM), owing to its robustness to multipath fading,  has been incorporated in many existing standards (e.g., IEEE $802.11$, IEEE $802.16$, DAB, DVB, HyperLAN, ADSL etc.) and is also a candidate for future wireless standards (e.g., IEEE $802.20$).
All current standards use pilot symbols to obtain channel state information (needed to perform coherent data detection).
This reduces the bandwidth available for data transmission, e.g., the IEEE $802.11$n standard uses $4$ subcarriers for pilots, that is $7.1\%$ of the available bandwidth, of the $56$ subcarriers available for transmission. Blind equalization methods are advantageous as they do not require regular training/pilots symbols, thus freeing up valuable bandwidth.

Several works exist in literature on blind channel estimation and equalization. A brief classification of these works based on a few commonly used constraints/assumptions is given in Table \ref{tableLitClass} (note that this list is not exhaustive).
Broadly speaking, the literature on blind channel estimation can be classified into maximum-likelihood (ML) methods and non-ML methods.

The non-ML methods include approaches based on subspace techniques \cite{IEEETVT:Shin2007}-\cite{IEEETSP:Muquet2002}, second-order statistics \cite{IEEETSP:Bolcskei2002}, \cite{IEEETSP:Gao2007}, \cite{IEEETSP:Muarkami2006}, Cholesky factorization \cite{IEEETSP:Choi2008}, iterative methods \cite{IEEETVT:Banani2010}, virtual carriers \cite{IEEETWC:Li2003} real signal characteristics \cite{IEEETVT:Gao2007} and linear precoding \cite{IEEETSP:Gao2007}, \cite{IEEETWC:Petropulu2004}. Subspace-based methods \cite{IEEETVT:Shin2007}-\cite{IEEETSP:Su2007}, \cite{IEEETSP:Zeng2004}-\cite{IEEETSP:Muquet2002} generally have lower complexity but suffer from slow convergence as they require many OFDM symbols to get an accurate estimate of the channel autocorrelation matrix. Blind methods based on second-order statistics \cite{IEEETSP:Bolcskei2002}, \cite{IEEETSP:Gao2007}, \cite{IEEETSP:Muarkami2006} also require the channel to be strictly stationary over several OFDM blocks. More often than not, this condition is not fulfilled in wireless scenarios (e.g., as in WLAN and fixed wireless applications). Methods based on Cholesky's factorization \cite{IEEETSP:Choi2008} and iterative techniques \cite{IEEETVT:Banani2010} suffer from high computational complexity.

Several ML-based blind methods have been proposed in literature \cite{IEEETSP:Sarmadi2009}, \cite{ICASSP:Weiyu2008}, \cite{IEEETSP:Ma2006}-\cite{IEEETSP:Yellin1993}, \cite{IEEETSP:Chang2008}. Although they incur a higher computational cost, their superior performance and faster convergence is very attractive. These characteristics make this class of algorithms suitable for block fading scenarios with short channel coherence time. Usually, suboptimal approximations are used to reduce the computational complexity of ML-based methods. Though these methods reduce the complexity of the exhaustive ML search, they still incur a significantly high computational cost. Some methods like \cite{IEEETSP:Ma2006}, \cite{IEEETSP:Larsson2003}, \cite{DSP:Stocia2003} are sensitive to initialization parameters, while others work only for specific constellations (see Table \ref{tableLitClass}). A few ML-based algorithms allow the channel to change on a symbol-by-symbol basis (e.g., \cite{IEEETSP:Naffouri2010, IEEETSP:Chang2008}), however, these algorithms are only able to deal with constant modulus constellations.

To the best of our knowledge no blind algorithm in literature is able to deal with channels that change from one OFDM symbol to another when the data symbols are drawn from a general constellation. Contrast this with the equalization algorithm presented in this paper. The key features of the blind equalization algorithm presented in this paper are that it
\begin{enumerate}
  \item works with an arbitrary constellation,
  \item can deal with channels that change for one symbol to the next,
  \item does not assume any statistical information about the channel.
\end{enumerate}
In addition, we propose a low-complexity implementation of the algorithm by utilizing the special structure of partial FFT matrices and prove that the complexity becomes especially low in the high SNR regime.

\begin{table*}
  \centering
  \caption{Literature Classification}\label{tableLitClass}
\begin{tabular}{|c|c|c|}
  \hline
  \textbf{Constraint} & Limited by & Not limited by\\
  \hline
  \hline
  & \cite{IEEETVT:Shin2007,IEEETSP:Moulines1995,IEEEWCNC:Tu2008,IEEETSP:Su2007,IEEETVT:Tu2010,IEEETSP:Zeng2004,IEEEJSAC:Gao2008,IEEETSP:Muquet2002}, & \\
  Channel constant over & \cite{IEEETSP:Bolcskei2002,IEEETSP:Gao2007,IEEETSP:Muarkami2006,IEEETSP:Choi2008,IEEETVT:Banani2010,IEEETWC:Li2003}, & \cite{IEEETSP:Naffouri2010,IEEETSP:Chang2008} \\
  $M$ symbols, $M>1$&\cite{IEEETVT:Gao2007,IEEETWC:Petropulu2004,IEEETSP:Sarmadi2009,IEEETSP:Ma2006,DSP:Stocia2003,IEEETC:Li2001,IEEETC:Cui2006} & \\
  \hline
  Uses pilots & \cite{IEEETVT:Shin2007,IEEETSP:Su2007,IEEEJSAC:Gao2008,IEEETSP:Muquet2002,IEEETSP:Bolcskei2002},
  & \\ 
  to resolve & 
  \cite{IEEETSP:Choi2008,IEEETVT:Banani2010,IEEETWC:Li2003,IEEETWC:Petropulu2004,IEEETSP:Sarmadi2009},
  & \cite{IEEETWC:Necker2004} \\
  phase  ambiguity &\cite{IEEETSP:Ma2006,IEEETC:Cui2006,IEEETSP:Ma2007b,IEEETSP:Naffouri2010,IEEETSP:Chang2008} & \\
  \hline
  &   \cite{IEEETSP:Moulines1995,IEEEWCNC:Tu2008,IEEETVT:Tu2010,IEEEJSAC:Gao2008,IEEETSP:Gao2007,IEEETSP:Muarkami2006}, & \\
  Constant modulus constellation & \cite{IEEETVT:Banani2010,IEEETWC:Li2003,IEEETSP:Sarmadi2009,IEEETSP:Ma2006}, & \cite{IEEETVT:Shin2007,IEEETSP:Su2007,IEEETSP:Zeng2004,IEEETSP:Bolcskei2002,IEEETSP:Choi2008}\\
  &   \cite{DSP:Stocia2003,IEEETSP:Naffouri2010,IEEETC:Li2001,IEEETWC:Necker2004,IEEETSP:Chang2008} &   \cite{IEEETVT:Gao2007,ICASSP:Weiyu2008,IEEETC:Cui2006} \\
  \hline
\end{tabular}
\end{table*}

The paper is organized as follows. Section \ref{System2} describes the system model and Section \ref{Desc2} describes the blind equalization algorithm. Section \ref{method2} presents an approximate method to reduce the computational complexity of the algorithm, while Section \ref{complexity2} evaluates this complexity in the high SNR regime. Section \ref{Sim2} presents the simulation results and Section \ref{con2} gives the concluding remarks.

\subsection{Notation}\label{notation2}

We denote scalars with small-case letters, $x$, vectors with small-case boldface letters, $\x$, while the individual entries of a vector $\she$ are denoted by $h(l)$. Upper case boldface letters, $\BX$, represent matrices while calligraphic notation, $\calX $, is reserved for vectors in the frequency domain. A hat over a variable indicates an estimate of the variable, e.g., $\hat{\h}$ is an estimate of $\h$. $(.)^{\rm T}$ and $(.)^{\rm H}$ denote the transpose and Hermitian operations, while the notation $\odot$ stands for element-by-element multiplication. The discrete Fourier transform (DFT) matrix is denoted by $\QQ$ and defined as ${q}_{l,k}={e^{-j\frac{2\pi}{N}(l-1)(k-1)}}$ with $k,l = 1,2,\cdots,N$ ($N$ is the number of subcarriers in the OFDM symbol), while the invrse DFT (IDFT) is denoted as $\QQ^{\rm H}$. The notation $\|\Ba\|^{2}_{\bf{B}}$ represents the weighted norm defined as  $\|\Ba\|^{2}_{\bf{B}}\stackrel{\Delta}{=} \Ba^{\rm H}{\mbox{\boldmath $B$}}\Ba$ for some vector $\Ba$ and matrix ${\mbox{\boldmath $B$}}$.

\section{System Model} \label{System2}

Consider an OFDM system where all the $N$ available subcarriers are modulated by data symbols chosen from an arbitrary constellation. The frequency-domain OFDM symbol $\calX$, of size $N \times 1$, undergoes an IDFT operation to produce the time-domain symbol $ {\x}$, i.e.
\begin{eqnarray}
\label{l'kdv'lkvbnfp}
{\x} = {\sqrt{N}}\QQ^{\rm H}  {\calX}.
\end{eqnarray}
The transmitter then appends a length $L$ cyclic prefix (CP) to $\x$ and transmits it over the channel. The channel $\h$, of maximum length $L + 1 < N$, is assumed to be constant for the duration of a single OFDM symbol, but could change from one symbol to the next. The received signal is a convolution of the transmitted signal with the channel observed in additive white circularly symmetric Gaussian noise ${\bf{n}} \sim \mathcal{N}(0,\BI)$. The CP converts the linear convolution relationship to circular convolution, which, in the frequency domain, reduces to an element-by-element operation. Discarding the CP, the frequency-domain received symbol is given by
\begin{eqnarray}
\label{iifhwr} \calY =  \sqrt{\rho}\; \calH \odot \calX + \calN,
\end{eqnarray}
where $\rho$ is the signal to noise ratio (SNR) and $\calY,$ $\calH,$ $\calX,$$\calN,$ are the $N$-point DFT's of $\y,$ $\h,$ $\x,$ and additive noise $\n$ respectively, i.e.
\begin{eqnarray}
\nonumber \label{'lkdgglblgg} \calH = \QQ\left[ \begin{array}{c} \she \\ {\bf{0}} \end{array}\right],\;\;\; \calX =
\frac{1}{\sqrt{N}}\QQ\x, \\\;\;\; \calN =
\frac{1}{\sqrt{N}}\QQ\n, \;
\mbox{and} \;\; \calY = \frac{1}{\sqrt{N}}\QQ\y.
\end{eqnarray}
Note that $\she$ is zero padded before taking its $N$-point DFT. Let $\A^{\rm H}$ consist of first $L+1$ columns of ${\QQ}$ (i.e., $\A$ consist of first $L+1$ rows of ${\QQ}^{\rm H}$), then
\begin{eqnarray}
\label{;OFpvdbf} \calH = \A^{\rm H}\h \;\;\; \mbox{and} \;\;\; \h = \A\calH.
\end{eqnarray}
This allows us to rewrite (\ref{iifhwr}) as
\begin{eqnarray}
\label{ifhwrr}  \calY =  \sqrt{\rho}\;\mbox{diag}(\calX)\A^{\rm H}\h + \calN.
\end{eqnarray}

\section{Blind Equalization Approach} \label{Desc2}

Consider the input/output equation (\ref{ifhwrr}), which in its element by element form reads
\begin{eqnarray}
\mathcal{Y}(j) = \sqrt{\rho}\;\mathcal{X}(j)\Ba_j^{\rm H}\h + \mathcal{N}(j)
\end{eqnarray}
where $\Ba_j$ is the $j$th column of $\BA$. The problem of joint ML channel estimation and data detection for OFDM channels can be cast as the following minimization problem
\begin{eqnarray}
J_{ML} &=&  \min_{h, \mathcal{X} \in \Omega^N} \|\calY - \sqrt{\rho}\; {\rm
diag}(\calX) \BA^{\rm H}\h\|^{2} \nonumber \\
&=& \min_{h, \mathcal{X} \in \Omega^N}
\sum_{i=1}^{N}|\mathcal{Y}(i) - \sqrt{\rho}\; \mathcal{X}(i)\Ba_i^{\rm H}\h|^{2} \nonumber \\ \nonumber
&=& \min_{h, \mathcal{X} \in \Omega^N} \left\{
\sum_{j=1}^{i}|\mathcal{Y}(j) - \sqrt{\rho}\;\mathcal{X}(j)\Ba_j^{\rm H}\h|^{2} + \right.\\ && \left.
\sum_{j=i+1}^{N}|\mathcal{Y}(j) - \sqrt{\rho}\;\mathcal{X}(j)\Ba_j^{\rm H}\h|^{2}
\right\} \label{obj}
\end{eqnarray}
where $\Omega^N$ denotes the set of all possible $N-$dimensional signal vectors. Let us consider a partial data sequence $\calX_{(i)}$ up to the time index $i$, i.e.\footnote{Thus, for example $\calX_{(2)}=[\mathcal{X}(1), \mathcal{X}(2)]^{\rm T}$ and $\calX_{(N)}=[\mathcal{X}(1), \cdots, \mathcal{X}(N)]^{\rm T} \stackrel{\Delta}{=} \calX$.}
\[
\calX_{(i)} = [\mathcal{X}(1) \;\; \mathcal{X}(2) \;\; \cdots \;\;
\mathcal{X}(i)]^{\rm T}
\]
and define $M_{\mathcal{X}_{(i)}}$ as the corresponding cost function, i.e.
\begin{equation}\label{Mx}
M_{\mathcal{X}_{(i)}} = \min_{{h}} \|\calY_{(i)} - \sqrt{\rho}\;{\rm
diag}(\calX_{(i)}) \BA_{(i)}^{\rm H}\h\|^{2}, \\
\end{equation}
where $\BA^{\rm H}_{(i)}$ consists of the first $i$ rows of $\BA^{\rm H}.$

In the following, we pursue an idea for blind equalization of single-input multiple-output systems first inspired by \cite{ICASSP:Weiyu2008}. Let $R$ be the optimal value for the objective function (\ref{obj}) (we show how to determine $R$ in Section \ref{Subsec:ChoiceOfR} further ahead). If $M_{\mathcal{X}_{(i)}} > R,$ then $\calX_{(i)}$ can not be the first $i$ symbols of the ML solution $\hat{\calX}^{\rm ML}$ to (\ref{obj}). To prove this, let $\hat{\calX}^{\rm ML}$ and $\hat{\h}^{\rm ML}$ denote the ML estimates and suppose that our estimate $\hat{\calX}_{(i)}$ satisfies
\begin{equation}
\label{XhateqXhatML}
\hat{\calX}_{(i)}=\hat{\calX}_{(i)}^{\rm ML}
\end{equation}
i.e. the estimate $\hat{\calX}_{(i)}$ matches the first $i$ elements of the ML estimate. Then, we can write
\begin{eqnarray}
\nonumber R &=& \min_{h, \mathcal{X} \in \Omega^N} \|\calY - \sqrt{\rho}\;{\rm
diag}(\calX) \BA^{\rm H}\h\|^{2} \\ \nonumber &=& \|\calY_{(i)} - \sqrt{\rho}\;{\rm diag} (\hat{\calX}_{(i)}^{\rm ML}) \BA_{(i)}^{\rm H} \hat{\h}^{\rm ML}\|^{2}\\ \nonumber && + \sum_{j=i+1}^{N}|\mathcal{Y}(j) - \sqrt{\rho}\;\hat{\mathcal{X}}^{\rm ML}(j) \Ba_j^{\rm H}\hat{\h}^{\rm ML}|^{2} \\ \nonumber &=& \|\calY_{(i)} - \sqrt{\rho}\;{\rm diag}(\hat{\calX}_{(i)}) \BA_{(i)}^{\rm H}\hat{\h}^{\rm ML}\|^{2} \\  && + \sum_{j=i+1}^{N}|\mathcal{Y}(j) - \sqrt{\rho}\;\hat{\mathcal{X}}^{\rm ML}(j)\Ba_j^{\rm H}\hat{\h}^{\rm ML}|^{2},
\end{eqnarray}
where the last equation follows from (\ref{XhateqXhatML}). Now, clearly
\begin{eqnarray}
\label{derivaR}
\nonumber \|\calY_{(i)} \hspace{-0.8em}&-& \hspace{-0.8em}\sqrt{\rho}\;{\rm diag}(\hat{\calX}_{(i)}) \BA_{(i)}^{\rm H}\hat{\h}^{\rm ML}\|^{2} \\ &\geq& \min_{{h}} \|\calY_{(i)} - \sqrt{\rho}\;{\rm diag}(\hat{\calX}_{(i)}) \BA_{(i)}^{\rm H}{\h}\|^{2} \\ &=& \|\calY_{(i)} - \sqrt{\rho}\;{\rm diag}(\hat{\calX}_{(i)}) \BA_{(i)}^{\rm H}{\hat{\h}}\|^{2},
\end{eqnarray}
where $\hat{\h}$ is the argument that minimizes the RHS of (\ref{derivaR}). Then
\begin{eqnarray}
\nonumber R &=&  \|\calY_{(i)} - \sqrt{\rho}\;{\rm diag}(\hat{\calX}_{(i)}) \BA_{(i)}^{\rm H}\hat{\h}^{\rm ML}\|^{2}\\ \nonumber && + \sum_{j=i+1}^{N}|\mathcal{Y}(j) - \sqrt{\rho}\;\hat{\mathcal{X}}(j)\Ba_j^{\rm H}\hat{\h}^{\rm ML}|^{2} \\ \nonumber &\ge&  \min_{{h}} \|\calY_{(i)} - \sqrt{\rho}\;{\rm diag}(\hat{\calX}_{(i)}) \BA_{(i)}^{\rm H}{\h}\|^{2} \\ &=& M_{\mathcal{X}_{(i)}}. \label{Mxhat}
\end{eqnarray}
So, for $\hat{\calX}_{(i)}$ to correspond to the first $i$ symbols of the ML solution $\hat{\calX}_{(i)}^{\rm ML}$, we should have $M_{\hat{\mathcal{X}}_{(i)}} < R$. Note that the above represents a necessary condition only. Thus if $\hat{\calX}_{(i)}$ is such that $M_{\hat{\mathcal{X}}_{(i)}} < R$, then this does not necessarily mean that $\hat{\calX}_{(i)}$ coincides with $\hat{\calX}_{(i)}^{\rm ML}$.

This suggests the following method for blind equalization. At each subcarrier frequency $i$, make a guess of the new value of $\mathcal{X}(i)$ and use that along with previous estimated values $\hat{\mathcal{X}}(1), ..., \hat{\mathcal{X}}(i-1)$ to construct $\hat{\calX}_{(i)}$. Estimate $\h$ so as to minimize $M_{\hat{\mathcal{X}}_{(i)}}$ in (\ref{Mxhat}) and calculate the resulting minimum value of $M_{\hat{\mathcal{X}}_{(i)}}$. If the value of $M_{\hat{\mathcal{X}}_{(i)}} < R$, then proceed to $i+1$. Otherwise, backtrack in some manner and change the guess of $\mathcal{X}(j)$ for some $j\leq i$. A problem with this approach is that for $i\leq L+1$, given any choice of $\hat{\mathcal{X}}(i)$, $\h$ can always be chosen by least-squares to make $M_{\hat{\mathcal{X}}_{(i)}}$ in (\ref{Mxhat}) equal to zero\footnote{Since $\BA^{\rm H}_{(i)}$ is full rank for $i \le L+1,$ ${\rm diag}(\calX_{(i)})\BA^{\rm H}_{(i)}$ is full rank too for each choice of ${\rm diag}(\calX_{(i)})$ and so we will always find some $\h$ that will make the objective function in (\ref{Mxhat}) zero (since $\h$ has $L+1$ degrees of freedom).}. Then, we will need at least $L+1$ pilots defying the blind nature of our algorithm. Alternatively, our search tree should be at least $L+1$ deep before we can obtain a nontrivial (i.e. nonzero) value for $M_{\hat{\mathcal{X}}_{(i)}}$.

An alternative strategy would be to find $\h$ using weighted regularized least squares. Specifically, instead of minimizing the objective function $J_{ML}$ in equation (\ref{obj}), we minimize the maximum a posteriori (MAP) objective function
\begin{eqnarray}\label{obj2}
J_{MAP} = \min_{h, \mathcal{X} \in \Omega^N} \left\{
\|\h\|^{2}_{R_h^{-1}} +\|\calY - \sqrt{\rho}\;{\rm
diag}(\calX) \BA^{\rm H}\h\|^{2}\right\}
\end{eqnarray}
where $\BR_h$ is the autocorrelation matrix of $\h$ (in Section \ref{method2}, we modify the blind algorithm to avoid the need for channel statistics). Now the objective function in (\ref{obj2}) can be decomposed as
\begin{eqnarray}\label{obj3}
\nonumber J_{MAP}\hspace{-0.8em} &=& \hspace{-0.8em}\min_{h, \mathcal{X} \in \Omega^N}   \left\{ \underbrace{
\|\h\|^{2}_{R_h^{-1}} 
+\sum_{j=1}^{i}|\mathcal{Y}(j) -
\sqrt{\rho}\;\mathcal{X}(j)\Ba_j^{\rm H}\h|^{2}}_{=M_{\mathcal{X}_{(i)}}} \right. \\  && + \left.
\sum_{j=1+1}^{N}|\mathcal{Y}(j) -\sqrt{\rho}\;
\mathcal{X}(j)\Ba_j^{\rm H}\h|^{2} \right\}
\end{eqnarray}
Given an estimate of $\hat{\calX}_{(i-1)}$, the cost function reads
\begin{eqnarray}\label{obj4}
\nonumber {M}_{\hat{\mathcal{X}}_{(i-1)}}\hspace{-1em} &=& \\
&&\hspace{-6.5em}  \min_{h}
\left\{\|\h\|^{2}_{R_h^{-1}}+\|\calY_{(i-1)} - \sqrt{\rho}\;
{\rm diag}(\hat{\calX}_{(i-1)}) \BA_{(i-1)}^{\rm H}\h\|^{2} \right\}
\end{eqnarray}
with the optimum value (see \cite{Book:Sayed2003}, Chapter $12$, pp. $671$)
\begin{eqnarray}
\nonumber \hat{\h} = \sqrt{\rho}\;\mbox{{\emph{\textbf{R}}}}_h\BA_{(i-1)}{\rm diag}(\hat{\calX}_{(i-1)}^{\rm H})\hspace{-1em} && \\ \nonumber &&\hspace{-14.5em} [\II+ \rho\;{\rm diag}(\hat{\calX}_{(i-1)}) \BA_{(i-1)}^{\rm H}\mbox{{\emph{\textbf{R}}}}_h \BA_{(i-1)}{\rm diag}(\hat{\calX}_{(i-1)}^{\rm H})]^{-1}\calY_{(i-1)}\\
\end{eqnarray}
and corresponding minimum cost (MMSE error)
\begin{equation}
\mbox{mmse}\hspace{-0.2em} =\hspace{-0.2em} [\mbox{{\emph{\textbf{R}}}}^{-1}_h + \rho\BA_{(i-1)}{\rm diag}(\hat{\calX}_{(i-1)})^{\rm H}{\rm diag}(\hat{\calX}_{(i-1)}) \BA_{(i-1)}^{\rm H}]^{-1}\label{eq:mmseOfMAPeq}
\end{equation}
If we have a guess of $\mathcal{X}(i)$, we can update the cost function and obtain $M_{\hat{\mathcal{X}}_{(i)}}$. In fact, the cost function $M_{\hat{\mathcal{X}}_{(i)}}$ is the same as that of $M_{\hat{\mathcal{X}}_{(i-1)}}$ with the additional observation $\mathcal Y(i)$ and an additional regressor $\hat{\mathcal{X}}(i)\Ba_i^{\rm H}$, i.e.
\begin{eqnarray}\label{obj4a}
\nonumber && M_{\hat{\mathcal{X}}_{(i)}} = \min_{h}
\left\{\|\h\|^{2}_{R_h^{-1}}+\right. \\ && \hspace{-2em}\left. \left\|\left[%
\begin{array}{c}
  \calY_{(i-1)} \\
  \mathcal{Y}(i) \\
\end{array}%
\right]
 -
\sqrt{\rho}\; \left[%
\begin{array}{c}
  {\rm diag}(\hat{\calX}_{(i-1)}) \BA_{(i-1)}^{\rm H} \\
  \hat{\mathcal{X}}(i)\Ba^{\rm H}_i \\
\end{array}%
\right]\h\right\|^{2} \right\}
\end{eqnarray}
We can thus, recursively update the value $M_{\hat{\mathcal{X}}_{(i)}}$ based on $M_{\hat{\mathcal{X}}_{(i-1)}}$ using recursive least squares (RLS) \cite{Book:Sayed2003}, i.e.
\begin{eqnarray}\label{cost}
{M}_{\hat{\mathcal{X}}_{(i)}} = {M}_{\hat{\mathcal{X}}_{(i-1)}}+
\gamma(i)|\mathcal{Y}(i)- \sqrt{\rho}\;\hat{\mathcal{X}}(i)\Ba^{\rm H}_i\hat{\h}_{i-1}|^{2}\label{J2}
\end{eqnarray}
\begin{eqnarray}\label{costh}
\hat{\h}_i = \hat{\h}_{i-1} +
\Bg_i\left(\mathcal{Y}(i)- \sqrt{\rho}\;\hat{\mathcal{X}}(i)\Ba^{\rm H}_i\hat{\h}_{i-1}\right)
\end{eqnarray}
where
\begin{eqnarray}
\Bg_i & = & \sqrt{\rho}\;\gamma(i)\hat{\mathcal{X}}(i)^{\rm H}\BP_{i-1}\Ba_i\label{g}\\
\gamma(i) & = &
\frac{1}{1+\rho|\hat{\mathcal{X}}(i)|^{2}\Ba^{\rm H}_i\BP_{i-1}\Ba_i}\label{gamma}\\
\BP_i & = & \BP_{i-1} -
\rho\;\gamma(i)|\hat{\mathcal{X}}(i)|^{2}\BP_{i-1}\Ba_i\Ba_i^{\rm H}\BP_{i-1}\label{Pi}
\end{eqnarray}
These recursions apply for all $i$ and are initialized by
\[
{M}_{\hat{\mathcal{X}}_{(-1)}} = 0, \;\;\; \BP_{-1} = \BR_h, \;\;\;
\mbox{and} \;\;\; \hat{\h}_{-1} = \vzero
\]
Now, let $R$ be the optimal value for the regularized objective function in (\ref{obj2}). If the value $R$ can be estimated, we can restrict the search of the blind MAP solution $\hat{\calX}$ to the offsprings of those partial sequences $\hat{\calX}_{(i)}$ that satisfy $M_{\hat{\mathcal{X}}_{(i)}} < R.$ This forms the basis for our exact blind algorithm described below.

%
%

\subsection{Exact Blind Algorithm} \label{blindalgo2}

In this subsection, we describe the algorithm used to find the MAP solution of the system. The algorithm employs the above set of iterations (\ref{cost})$-$(\ref{Pi}) to update the value of the cost function $M_{\hat{\mathcal{X}}_{(i)}}$ which is then compared with the optimal value $R$. The input parameters for the algorithm are: the received channel output $\calY$, the initial search radius $r$, the modulation constellation\footnote{Examples of the modulation constellation are $\Omega$ are $4$-QAM and $16$-QAM. We use $|\Omega|$ to denote the constellation size and $\Omega(k)$ for the $k$th constellation point. For example, in $4$-QAM $|\Omega| = 4$ and $\Omega(1),\cdots, \Omega(4)$ are the four constellation points of $4$-QAM. The indicator $I(i)$ refers to the last constellation point visited by our search algorithm at the $i$th subcarrier.} $\Omega$ and the $1\times$$N$ index vector $I$.

The algorithm is described as follows (the algorithm is also described by the flowchart in Figure \ref{BlinAFC})




\begin{enumerate}
  \item (\textbf{Initialize}) Set $i=1$, $I(i)=1$ and set $\hat{\mathcal{X}}(i)=\Omega(I(i)).$
  \item (\textbf{Compare with bound}) Compute and store the metric $M_{\hat{\mathcal{X}}_{(i)}}.$ If $M_{\hat{\mathcal{X}}_{(i)}}>r,$ go to 3; else, go to 4;
  \item (\textbf{Backtrack})\; Find the largest\;\; $1\leq$$j\leq$$i$ \;\; such that\\ $I(j)<|\Omega|$.  If there exists such $j,$ set $i=j$ and go to 5; else go to 6.
  \item (\textbf{Increment subcarrier}) If $i<N$ set $i=i+1, I(i)=1$, $\hat{\mathcal{X}}(i)=\Omega(I(i))$ and go to 2; else store current $\hat{\calX}_{(N)},$ update $r=M_{\hat{\mathcal{X}}_{(N)}}$ and go to 3.
  \item (\textbf{Increment constellation}) Set $I(i)=I(i)+1$ and $\hat{\mathcal{X}}(i)=\Omega(I(i)).$ Go to 2.
  \item (\textbf{End/Restart}) If a full-length sequence $\hat{\calX}_{(N)}$ has been found in Step 4, output it as the MAP solution and terminate; otherwise, double $r$ and go to 1.
\end{enumerate}

The essence of the algorithm is to eliminate any choice of the input that increments the objective function beyond the radius $r$. When such a case is confronted, the algorithm backtracks (Step $3$ then Step $5$) to the nearest subcarrier whose alphabet has not been exhausted (the nearest subcarrier will be the current subcarrier if its alphabet set is not exhausted).

The other dimension the algorithm works on is properly sizing $r$; if $r$ is too small such that we are not able to backtrack, the algorithm doubles $r$ (Step $3$ then Step $6$). If on the other hand $r$ is too large that we reach the last subcarrier too fast, the algorithm reduces $r$ to the most recent value of the objective function. ($r = M_{\mathcal{X}_{(N)}}$) and backtracks (Step $4$ then Step $3$).
\begin{figure*}[htb]
\begin{center}
\epsfxsize = 7 true in \epsfbox{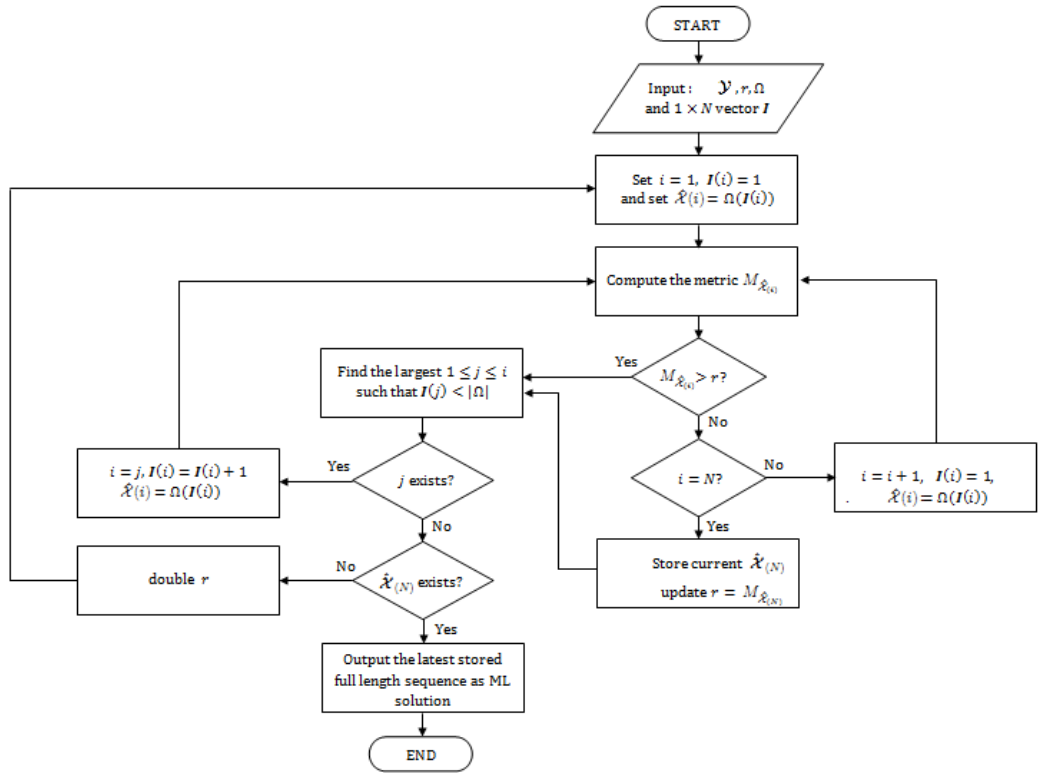}
\caption{\small Flowchart of the blind algorithm.} \label{BlinAFC}
\end{center} \end{figure*}

\vspace{0.5em}
\noindent {\emph{Remark 1:} 
The backtracking algorithm depends heavily on calculating the cost function using (\ref{cost})-(\ref{Pi}). In the constant modulus case, the values of $\rho|\hat{\mathcal{X}}(i)|^{2}$ in equations (\ref{gamma}) and (\ref{Pi}) become constant (equal to $\rho\;\mathcal{E}_\mathcal{X}$) for all $i$, and the values of $\gamma(i)$ and $\BP_i$ become
\begin{eqnarray}
\gamma(i) & = &
\frac{1}{1+\rho\;\mathcal{E}_\mathcal{X}\Ba^{\rm H}_i\BP_{i-1}\Ba_i}\\
\BP_i & = & \BP_{i-1} -
\rho\;\mathcal{E}_\mathcal{X}\gamma(i)\BP_{i-1}\Ba_i\Ba_i^{\rm H}\BP_{i-1},
\end{eqnarray}
which are independent of the transmitted signal and thus can be calculated offline.

\vspace{0.5em}
\noindent {\emph{Remark 2:} The algorithm can also be used for a pilot-based standard. In this case, when the algorithm reaches a pilot holding-subcarrier, no backtracking is performed as the value of the data carrier is known perfectly. In the presence of pilots, it is wise to execute the algorithms over the pilot-holding subcarriers first and subsequently move to the data subcarriers. For equispaced comb-type pilots, (semi)-orthogonality of regressors is still guaranteed.

\vspace{0.5em}
\noindent {\emph{Remark 3:} Like all blind algorithms, we use one pilot bit to resolve the sign ambiguity (see references in Table \ref{tableLitClass}).
\subsection{Determination of initial radius $\rho$, $\BR_h$ and $r$} 
\label{Subsec:ChoiceOfR}
Our algorithm depends on $\rho$, $\BR_h$ and $r$ which we need to determine. The receiver can easily estimate $\rho$ by measuring the additive noise variance at its side. As for the channel covariance matrix $\BR_h$, our simulations show that with carrier reordering we can replace $\BR_h$ with identity with essentially no effect on performance. This becomes especially true in the high SNR regime. It remains to obtain an initial guess of the search radius $r$. To this emd, note that if $\h$ and $\calX$ are perfectly known (with $\h$ drawn from $\mathcal{N}({\bf{0}},\BR_h)$ but is known) then
\begin{equation}\label{obj4}
\xi = \|\h\|^{2}_{R_h^{-1}}+ \|\calY - \sqrt{\rho}\;{\rm
diag}(\calX) \BA^{\rm H}\h\|^{2}
\end{equation}
is a chi-square random variable with $k=2(N+L+1)$ degrees of freedom\footnote{The first term on the right hand side has $2(L+1)$ degrees of freedom as $\h$ is Gaussian distributed while the second term has $2N$ degrees of freedom as $\calY - \sqrt{\rho}\;{\rm diag}(\calX) \BA^{\rm H}\h$ is just Gaussian noise.}. Thus, the search radius should be chosen such that $P(\xi > r)\leq{\epsilon}$, where $P(\xi > r) = 1 - F(r;\,k)$, and where $F(r;\,k)$ is the cumulative distribution function of the chi-square random variable given by
\begin{equation}
F(r;\,k) = \frac{\gamma(k/2,\,r/2)}{\Gamma(k/2)},
\end{equation}
Here, $\gamma(k/2,\,r/2)$ is the lower incomplete gamma function defined as
\begin{equation}
\gamma(k/2,\,r/2) = \int_0^{r/2} t^{k/2-1}\,e^{-t}\,{\rm d}t.
\end{equation}
So, under this initial radius, we guarantee finding the MAP solution with probability at least $1-\epsilon$. In case a solution is not found, the algorithm doubles the value of $r$ and starts over. This process continues until a solution is found. For example, when $N=64, L=15$ and $\epsilon = 0.01$, the value of our radius should be set to $204$. 


\section{An Approximate Blind Equalization Method} \label{method2}

There are two main sources that contribute to the complexity of the exact blind algorithm of Section \ref{Desc2}:
\begin{enumerate}
    \item \emph{{Calculating $\BP_i$:}} the second step of the blind algorithm requires updating the metric $M_{\hat{\mathcal{X}}_{(N)}}$. This metric depends heavily on operations involving the $(L+1) \times (L+1)$ matrix $\BP_i$ which are the most computationally expansive (see Table \ref{tableRLS} which estimates the computational complexity of the RLS).
    \item \emph{{Backtracking:}} When the condition $M_{\hat{\mathcal{X}}_{(i)}}\leq r$ is not satisfied, we need to backtrack and pursue another branch of the search tree. This represents a major source of complexity.
\end{enumerate}
In the following, we show how we can avoid calculating $\BP_i$ all together. We postpone the issue of backtracking to Section \ref{complexity2}.

\subsection{Avoiding $\BP_i$}
Note that in the RLS recursions (\ref{J2})$-$(\ref{Pi}), $\BP_i$ always appears multiplied by $\Ba_i$. Let's see how this changes if we set $\BP_{-1}=\BI$ and assume that the $\Ba_i$'s are orthogonal or, in particular, if we assume that $\Ba^{\rm H}_i\Ba_{i+1}=\Ba^{\rm H}_i\Ba_{i+2}=0.$ With these assumptions note that
\begin{equation}
\gamma(0) =
\frac{1}{1+\rho\;|\hat{\mathcal{X}}(0)|^{2}\Ba^{\rm H}_0 \BP_{-1}\Ba_0}= \frac{1}{1+\rho\; |\hat{\mathcal{X}}(0)|^{2}(L+1)}
\end{equation}
i.e., $\gamma(0)$ is independent of $\BP_{-1}$. Also note that
\begin{eqnarray}
\nonumber \BP_0\Ba_1 &=& \BP_{-1}\Ba_1 -\rho\;
\gamma(0)|\hat{\mathcal{X}}(0)|^{2}\BP_{-1}\Ba_0\Ba_0^{\rm H}\BP_{-1}\Ba_1\\ \nonumber &=&
 \Ba_1 -\rho \;
\gamma(0)|\hat{\mathcal{X}}(0)|^{2}\Ba_0\Ba_0^{\rm H}\Ba_1\\
&=&  \Ba_1. \label{a1}
\end{eqnarray}
For a similar reason
\begin{equation}\label{a2}
\BP_0\Ba_2 = \Ba_2.
\end{equation}
From (\ref{a1}), it is also easy to conclude that
\begin{equation}
\gamma(1) =
\frac{1}{1+\rho\;|\hat{\mathcal{X}}(1)|^{2}(L+1)}
\end{equation}
i.e., $\gamma(1)$ is independent of $\BP_0$. Also, from (\ref{a1}) and (\ref{a2}) it follows that $\BP_i \Ba_{i+1}=\Ba_{i+1}$ and $\BP_i \Ba_{i+2}=\Ba_{i+2}$. We now investigate what happens to $\BP_{i+1}$.
\begin{eqnarray}
\nonumber  && \hspace{-2em} \BP_{i+1}\Ba_{i+2}  \\ \nonumber  &=& \BP_{i}\Ba_{i+2} -\rho\;
\gamma({i+1})|\hat{\mathcal{X}}({i+1})|^{2}\BP_{i}\Ba_{i+1}\Ba_{i+1}^{\rm H}\BP_{i}\Ba_{i+2}\\ \nonumber
& =&  \Ba_{i+2} -\rho \;
\gamma({i+1})|\hat{\mathcal{X}}({i+1})|^{2}\Ba_{i+1}\Ba_{i+1}^{\rm H}\Ba_{i+2}\\
& =&  \Ba_{i+2}.
\end{eqnarray}
Similarly,
\begin{equation}
\BP_{i+1}\Ba_{i+3} = \Ba_{i+3}.
\end{equation}
So, by induction we see that each occurrence of $\BP_i \Ba_i$ in the recursion set (\ref{cost})-(\ref{gamma}) can be replaced with $\Ba_i$. This allows us to discard (\ref{Pi}), i.e.,
\begin{eqnarray}\label{cost3}
{M}_{\hat{\mathcal{X}}_{(i)}} &=& {M}_{\hat{\mathcal{X}}_{(i-1)}}+
\gamma(i)|\mathcal{Y}(i)- \sqrt{\rho}\;\hat{\mathcal{X}}(i)\Ba^{\rm H}_i\hat{\h}_{i-1}|^{2}\label{J23} \\
\hat{\h}_i &=& \hat{\h}_{i-1} + \Bg_i\left(\mathcal{Y}(i)
- \sqrt{\rho}\;\hat{\mathcal{X}}(i)\Ba^{\rm H}_i\hat{\h}_{i-1}\right),
\end{eqnarray}
where
\begin{eqnarray}
\Bg_i & = & \sqrt{\rho}\;\gamma(i)\hat{\mathcal{X}}(i)^{\rm H}\Ba_i\label{g3}\\
\gamma(i) & = &
\frac{1}{1+\rho\;|\hat{\mathcal{X}}(i)|^{2}(L+1)}.\label{gamma3}
\end{eqnarray}
Thus, the approximate blind RLS algorithm is effectively running at LMS complexity. Table~\ref{tableRLS} summarizes the computational complexity incurred in the RLS calculation.

\begin{table}
  \centering
  \caption{Estimated computational cost per iteration of the RLS algorithm }\label{tableRLS}
  \begin{tabular}{|c|c|c|c|}
    \hline
    \textbf{Term} & $\times$ & $+$ & $\div$\\
    \hline
    \hline
    $\sqrt{\rho}\;\hat{\mathcal{X}}(i)\Ba^{\rm H}_i\hat{\h}_{i-1}$ & $2L+2$ & $L$ &  \\
    $|\mathcal{Y}(i)-\sqrt{\rho}\;\hat{\mathcal{X}}(i)\Ba^{\rm H}_i\hat{\h}_{i-1}|^{2}$ & 1 & 1 &  \\
    $\rho\;\gamma(i)$ & 1 &  & 1 \\
    $M_{\hat{\mathcal{X}}_{(i)}}$ & 1 & 1 &  \\
    $\hat{\h}_i$ & $L+2$ & $L+1$ & 1 \\
    $\BP_{i-1}\Ba_i$ & $L^2+2L+1$ & $L^2+L$ &  \\
    $\Bg_i$ & $L+3$ &  &  \\
    $\Ba^{\rm H}_i\BP_{i-1}\Ba_i$ & $L+1$ & $L$ &  \\
    $\gamma(i)$ & 3 & 1 & 1 \\
    $\Ba^{\rm H}_i\BP_{i-1}$ & $L^2+2L+1$ & $L^2+L$ &  \\
    $\BP_i$ & $L^2+2L+2$ & $L^2+2L+1$ &  \\
    \hline
    \hline
    \textbf{Total per iteration} & $3L^2+11L+17$ & $2L^2+5L+4$ & 3 \\
    \hline
  \end{tabular}
\end{table}

\subsection{Avoiding $\BP_i$ with Carrier Reordering}

The reduction in complexity above is based on two assumptions. The first assumption is to set $\BP_{-1}=\BI$ (instead of $\BR_h$) and the second is to assume that the consecutive $\Ba_i$'s are orthogonal. Note that the $\Ba_i$'s are columns of $\A$, i.e. they are partial FFT vectors. As such, strictly speaking, they are not orthogonal. Notice, however, that for $i\neq i'$,
\begin{equation}
\Ba_i^{\rm H}\Ba_{i'} = \sum^{L}_{k=0}e^{(j\frac{2\pi}{N}(i-i')k)},
\end{equation}
which after straightforward manipulation can be shown to be
\begin{eqnarray}
|\Ba_i^{\rm H}\Ba_{i'}| =\left\{
\begin{array}{cc}
  L+1, & (i=i') \\
  \frac{1}{L+1}\left|\frac{sin(\pi(i-i')\frac{L+1}{N})}{sin(\pi(i-i')\frac{1}{N})}\right|, & (i\neq i') \\
\end{array}\right.\label{eq:eq41}
\end{eqnarray}
This is a function of $(i-i')$ mod $N$. Thus, without loss of generality, we can set $i'=1$ and plot this autocorrelation with respect to $i$. The autocorrelation decays with $i$ as shown in Figure \ref{autocorr}.
\begin{figure}[htb]
\begin{center}
\epsfxsize = 3.5 true in \epsfbox{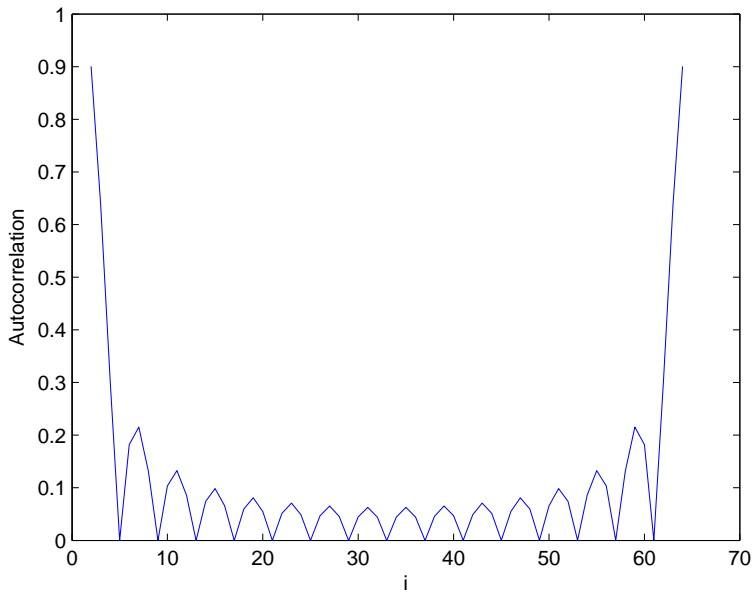} \caption{\small
Autocorrelation vs $i$ for $N=64$ and $L=15$} \label{autocorr}
\end{center} \end{figure}
We can use this observation in implementing our blind RLS algorithm. Specifically, note that the whole OFDM data is available to us and so we can visit the data subcarriers in any order we wish. The discussion above shows that the data subcarriers should be visited in the order $\;i, \;i+\Delta, \;i+2\Delta, \ldots$ where $\Delta$ should be chosen as large as possible to make $\;\Ba_i, \;\Ba_{i+\Delta}, \;\Ba_{i+2\Delta}, \ldots$ as orthogonal as possible, but small enough to avoid revisiting (or looping back to) a neighborhood too early. We found the choice $\Delta = \frac{N}{L+1}$ 
to be a good compromise. From Figure \ref{autocorr}, which plots (\ref{eq:eq41}) for $N=64$ and $L=15$, columns $1, 5, 9, 13, 17, 21,\cdots, 61$ are orthogonal to each other and so are the columns $2, 6, 10, 14, 18, \cdots, 62$. So, if the vectors are visited in the following order $1, 5, 9, 13, 17, 21, \cdots, 61, 2, 6, 10, 14, 18, \cdots, 62,\cdots$, then we have a consecutive set of vectors that are orthogonal. The only exception is in going from column $61$ to $2$. These two columns are not really orthogonal but are nearly orthogonal (the correlation of columns $1$ and $61$ is zero, so the correlation of 61 with 2 should be very small since the correlation function is continuous as shown in Figure \ref{autocorr}). In general, we chose $\Delta = \frac{N}{L+1}$ and visit the columns in the order ${i+\Delta, i+2\Delta, \cdots, i+L\Delta, i=1, \cdots, \Delta -1}$.

Our simulation results show that the BER we get with exact calculation of $\BP_i$ and that obtained when we set $\BP_{-1}=\BI$ with subcarrier reordering are almost the same. Table~\ref{tableRLSwithout} gives the computational complexity incurred in the RLS calculation when subcarrier reordering is used (i.e., free from $\BP_i$ calculation).
\begin{table}
  \centering
  \caption{Estimated computational cost per iteration of the RLS algorithm with Carrier Reordering}\label{tableRLSwithout}
  \begin{tabular}{|c|c|c|c|}
    \hline
    \textbf{Term} & $\times$ & $+$ & $\div$\\
    \hline
    \hline
    $\sqrt{\rho}\;\hat{\mathcal{X}}(i)\Ba^{\rm H}_i\hat{\h}_{i-1}$ & $2L+2$ & $L$ &  \\
    $|\mathcal{Y}(i)-\sqrt{\rho}\;\hat{\mathcal{X}}(i)\Ba^{\rm H}_i\hat{\h}_{i-1}|^{2}$ & 1 & 1 &  \\
    $\rho\;\gamma(i)$ & 1 &  & 1 \\
    $M_{\hat{\mathcal{X}}_{(i)}}$ & 1 & 1 &  \\
    $\hat{\h}_i$ & $L+2$ & $L+1$ & 1 \\
    $\gamma(i)$ & 3 & 1 & 1 \\
    \hline
    \hline
    \textbf{Total per iteration} & $4L+13$ & $2L+4$ & 3 \\
    \hline
  \end{tabular}
\end{table}

Note that with subcarrier reordering, the new version of the RLS runs without the need to use the power delay profile statistics, which relieves us from the need to provide this information.

\section{Computational Complexity in the High SNR Regime}\label{complexity2}

In the section, we study the other source of complexity (backtracking) and show that there is almost no backtracking\footnote{The term "backtracking" refers to the case when the algorithm is currently at subcarrier $i$ and it has to change the estimate of the data symbol at some subcarrier $j<i$. On the other hand, sweeping the constellation points at subcarrier to find the first one that satisfies ${M}_{{\mathcal{X}}_{(i)}} \leq r$ is not considered backtracking.} in the high SNR regime. To this end, consider the behavior of the algorithm when processing the $i$th subcarrier. There are $|\Omega|$ different alphabet possibilities to choose from at this subcarrier and a similar number of possibilities at the preceding $i-1$ subcarriers, creating a total of $|\Omega|^i - 1$ incorrect sequences $\bar{\calX}_{(i)}$ and one correct sequence $\hat{\calX}_{(i)}$. The best case scenario is to have only one sequence that satisfies ${M}_{\bar{\mathcal{X}}_{(i)}} \leq r$ in which case there would be only one node to visit. The worst case is having to visit the remaining $|\Omega|^i - 1$ wrong nodes before reaching the true sequence (visiting of nodes will happen through backtracking); this latter case is equivalent to the exhaustive search scenario (i.e., all possible sequences satisfy ${M}_{\bar{\mathcal{X}}_{(i)}} \leq r$). Thus, if we let $C_i$ denote the expected number of nodes visited at the $i$th subcarrier, then from above we can write
\begin{equation}\label{Cii}
C_i \leq 1+(|\Omega|^i - 1)P_i
\end{equation}
where $P_i$ is the maximum probability that an erroneous sequence of symbols $\bar{\calX}_{(i)}\neq\hat{\calX}_{(i)}$ has a cost less than $r$. We will show that this probability becomes negligibly small at high SNR values. Recall that
\begin{eqnarray}\label{Yii}
\calY_{(i)} =  \sqrt{\rho}\;\mbox{diag}(\hat{\calX}_{(i)})\A_{(i)}^{\rm H} \h + \calN_{(i)}
\end{eqnarray}
where $\calN_{(i)}$ denotes the first $i$ symbols of $\calN$. Note the (\ref{Yii}) can be written as
\begin{equation}\label{Yii3}
\calY_{(i)} = \begin{bmatrix}
\sqrt{\rho}\;{\rm diag}(\hat{\calX}_{(i)}) \BA^{\rm H}_{(i)} &
\BI \\
\end{bmatrix}
\begin{bmatrix}
\h \\
\calN_{(i)} \\
\end{bmatrix}
\end{equation}
We first prove our claim for the least squares (LS) cost and then show how the MAP cost reduces to LS cost for high SNR.

\subsection{LS cost}

Suppose we have an erroneous sequence of symbols $\bar{\calX}_{(i)}\neq\hat{\calX}_{(i)}$. The LS estimate of $\h$ is found by minimizing the objective function
\begin{eqnarray}\label{obj2LS}
J_{LS} = \min_{h, \mathcal{X} \in \Omega^N} \left\{
\|\calY_{(i)} - \sqrt{\rho}\;{\rm
diag}(\calX_{(i)}) \BA_{(i)}^{\rm H}\h\|^{2}\right\}
\end{eqnarray}
and the solution of $\h$ is (see \cite{Book:Sayed2003}, Chapter $12$, pp. $664$)
\begin{equation}
\hat{\h} = [\BA_{(i)}{\rm diag}(\hat{\calX}_{(i)}^{\rm H}){\rm diag}(\hat{\calX}_{(i)})\BA_{(i)}^{\rm H}]^{-1}   \sqrt{\rho}\;\BA_{(i)}{\rm diag}(\hat{\calX}_{(i)}^{\rm H})\calY_{(i)}.
\end{equation}
The cost associated with the LS solution is given by (see \cite{Book:Sayed2003}, Chapter $11$, pp. $663$)
\begin{eqnarray}
\nonumber {M}_{\bar{\mathcal{X}}_{(i)}}\hspace{-1em} &=&\hspace{-0.75em} \calY_{(i)}^{\rm H}\Big(\mbox{{\emph{\textbf{I}}}} - \sqrt{\rho}\;{\rm diag}(\bar{\calX}_{(i)}) \BA_{(i)}^{\rm H} \left[\sqrt{\rho}\;\BA_{(i)}{\rm diag}(\bar{\calX}_{(i)})^{\rm H} \right. \\ \nonumber && \left.\sqrt{\rho}\;{\rm diag}(\bar{\calX}_{(i)}) \BA_{(i)}^{\rm H}\right]^{-1} \sqrt{\rho}\BA_{(i)}{\rm diag}(\bar{\calX}_{(i)}^{\rm H})\Big) \calY_{(i)} \\ \nonumber
&=& \hspace{-0.75em}\calY_{(i)}^{\rm H}\Big(\mbox{{\emph{\textbf{I}}}} - \rho\;{\rm diag}(\bar{\calX}_{(i)}) \BA_{(i)}^{\rm H} \left(\rho\BA_{(i)}|{\rm diag}(\bar{\calX}_{(i)})|^2\BA_{(i)}^{\rm H}\right)^{-1} \\ \nonumber && \BA_{(i)}{\rm diag}(\bar{\calX}_{(i)}^{\rm H})\Big) \calY_{(i)} \\ \nonumber
&=& \hspace{-0.75em}\calY_{(i)}^{\rm H}\Big(\mbox{{\emph{\textbf{I}}}} - \frac{\rho}{\rho}\BD\Big) \calY_{(i)} \\
{M}_{\bar{\mathcal{X}}_{(i)}} &=& \calY_{(i)}^{\rm H}\Big(\mbox{{\emph{\textbf{I}}}} - \BD\Big) \calY_{(i)}
\end{eqnarray}
where
\begin{equation}
\BD = {\rm diag}(\bar{\calX}_{(i)}) \BA_{(i)}^{\rm H} \left(\BA_{(i)}|{\rm diag}(\bar{\calX}_{(i)})|^2\BA_{(i)}^{\rm H}\right)^{-1} \hspace{-1em}\BA_{(i)}{\rm diag}(\bar{\calX}_{(i)}^{\rm H}).
\end{equation}
So the probability that the sequence $\bar{\calX}_{(i)}$ satisfies ${M}_{\bar{\mathcal{X}}_{(i)}} \leq r$ reads
\begin{eqnarray}
\nonumber P_i &=&\mbox{Pr} ({M}_{\bar{\mathcal{X}}_{(i)}} \leq r) \\
P_i &=& \mbox{Pr} \bigg( \calY_{(i)}^{\rm H}\Big(\mbox{{\emph{\textbf{I}}}} - \BD\Big) \calY_{(i)} \leq r\bigg)\label{P_i_inter1LS}
\end{eqnarray}
In the strict sense of the word, backtracking means visiting Step $3$ in our algorithm. Substituting (\ref{Yii3}) in (\ref{P_i_inter1LS}) yields
\begin{equation}\label{PPiiLS}
P_i= {\rm Pr}\left(\left(
\begin{bmatrix}
  \h \\
  \calN_{(i)} \\
\end{bmatrix}^{\rm H} \BG_{(i)}  \begin{bmatrix}
  \h \\
  \calN_{(i)} \\
\end{bmatrix} \right)\leq r\right)
\end{equation}
where
\begin{equation}
\BG_{(i)} \hspace{-0.3em}= \hspace{-0.3em}\begin{bmatrix}
  \sqrt{\rho}\;\BA_{(i)}{\rm
diag}(\hat{\calX}_{(i)}^{\rm H})  \\
  \BI \\
\end{bmatrix}\hspace{-0.3em} \left[\BI - \BD \right]\hspace{-0.3em} \begin{bmatrix}
   \sqrt{\rho}\;{\rm diag}(\hat{\calX}_{(i)}) \BA^{\rm H}_{(i)} &
  \hspace{-0.2em} \BI \\
\end{bmatrix}.\label{eq:matrixGLS}
\end{equation}
Let $\BB = {\rm diag}(\hat{\calX}_{(i)}) \BA^{\rm H}_{(i)}$, then $\BG_{(i)}$ can be written as
\begin{equation}
\BG_{(i)} = \begin{bmatrix}
  \rho\; \BB^{\rm H}\left[\BI - \BD \right] \BB  & \BB^{\rm H}\left[\BI - \BD \right]\BI\\
  \BI\left[\BI - \BD \right] \BB & \BI\left[\BI - \BD \right]\BI\\
\end{bmatrix}\label{eq:matrixGLS2}
\end{equation}
which in compact form can be expressed as
\begin{equation}
\BG_{(i)} = \begin{bmatrix}
  \rho \BE  & \BE_2\\
  \BE_2^{\rm H} & \BE_3
\end{bmatrix}.\label{eq:matrixGLS3}
\end{equation}
Using the Chernoff bound the right hand side of (\ref{PPiiLS}) can be bounded in the following way
\begin{eqnarray}\label{PiminLS}
P_i \leq  e^{\mu r} E \Bigg[\mbox{exp}\left({-\mu 
\begin{bmatrix}
  \h \\
  \calN_{(i)}
\end{bmatrix}^{\rm H} \BG_{(i)}  \begin{bmatrix}
  \h \\
  \calN_{(i)}
\end{bmatrix} }\right) \Bigg].
\end{eqnarray}
Noting that
\begin{equation}
 \begin{bmatrix}
   \h  \\\calN_{(i)}
 \end{bmatrix} \sim \mathcal{N}(0, \BSigma_{(i)})
\end{equation}
with
\begin{equation}
\BSigma_{(i)} =   \begin{bmatrix} \BR_h & \bf{0} \\ \bf{0} & \BI_i \end{bmatrix},
\end{equation}
we can solve the expression in (\ref{PiminLS}) as
\begin{eqnarray}
\nonumber P_i \hspace{-1em} &\leq& \hspace{-1em} \frac{\mathlarger{\int}\mbox{exp}{\left(-\mu
\begin{bmatrix}
  \h \\
  \calN_{(i)}
\end{bmatrix}^{\rm H} \BG_{(i)}  \begin{bmatrix}
  \h \\
  \calN_{(i)}
\end{bmatrix} \right)}}{e^{-\mu r} \pi
^{(L+i+1)}}\\ \nonumber
&& \mbox{exp}{\left(-
\begin{bmatrix}
  \h \\
  \calN_{(i)}
\end{bmatrix}^{\rm H}  \BSigma_{(i)} \begin{bmatrix}
  \h \\
  \calN_{(i)}
\end{bmatrix} \right)}d\h d\calN_{(i)}\\ \nonumber
&=&\hspace{-1em} \frac{\mathlarger{\int}\mbox{exp}{\left(-
\begin{bmatrix}
  \h \\
  \calN_{(i)}
\end{bmatrix}^{\rm H} (\BSigma_{(i)} + \mu \BG_{(i)})  \begin{bmatrix}
  \h \\
  \calN_{(i)}
\end{bmatrix} \right)} d\h d\calN_{(i)}}{e^{-\mu r} \pi
^{(L+i+1)}}\\
&=&\hspace{-1em} \frac{\mathlarger{\int}\mbox{exp}\left({-\left|\left|
  \begin{bmatrix}
  \h \\
  \calN_{(i)}
\end{bmatrix} \right|\right|^2_{(\BSigma_{(i)} + \mu \BG_{(i)})}} \right) d\h d\calN_{(i)}}{e^{-\mu r} \pi
^{(L+i+1)}}.\label{eq:gaussianIntegralLS}
\end{eqnarray}
Note that the numerator in (\ref{eq:gaussianIntegralLS}) is a multi-variate complex Gaussian integral. Recall that an $n$-dimensional complex Gaussian integral has the solution (see \cite{ICASSP:Weiyu2008})
\begin{equation}
  \mathlarger{\int} \mbox{exp}\left({-\left|\left| \bf{x} \right|\right|^2_{\BW}}\right) d{\bf{x}}= \frac{\pi^n}{{\rm det} (\BW)}.
\end{equation}
This allows us to simplify (\ref{eq:gaussianIntegralLS}) as
\begin{equation}
P_i \leq \frac{e^{\mu r}}{{\rm det} (\BSigma_{(i)} + \mu
\BG_{(i)})}. \label{eq:P_ISimplfiedLS}
\end{equation}
Next, we show that the probability $P_i \rightarrow 0$ as $\rho \rightarrow \infty$. To show this, we just need to show that the largest eigenvalue of the term in the denominator goes to infinity as $\rho \rightarrow \infty$.

\begin{lemma}
Let $\BE=\BA_{(i)} \mathrm{diag}(\hat{\calX}_{(i)}^{\rm H})[\BI-\BD]\mathrm{diag}(\hat{\calX}_{(i)}) \BA_{(i)}^{\rm H}$ be a $(L+1)\times (L+1)$ matrix, then for any sequence $\hat{\calX}_{(i)}$, $\BE$ has a positive maximum eigenvalue, $\lambda_{\rm max}$ and a corresponding unit-norm eigenvector $\bf{v}$ of size $(L+1)\times 1$.
\end{lemma}

\begin{proof}
Recall that
\begin{eqnarray}
\nonumber \BD &=& {\rm diag}(\bar{\calX}_{(i)}) \BA_{(i)}^{\rm H} \left(\BA_{(i)}{\rm diag}(\bar{\calX}_{(i)}^{\rm H}){\rm diag}(\bar{\calX}_{(i)})\BA_{(i)}^{\rm H}\right)^{-1} \\ && \BA_{(i)}{\rm diag}(\bar{\calX}_{(i)}^{\rm H})
\end{eqnarray}
and let $\BF = {\rm diag}(\bar{\calX}_{(i)})\BA_{(i)}^{\rm H}$, then we can write the above equation as
\begin{equation}
 \BD = \BF\left(\BF^{\rm H} \BF\right)^{-1}  \BF^{\rm H} = \BF \BF^{\dag}
\end{equation}
where $\BF^{\dag}=\left(\BF^{\rm H} \BF\right)^{-1}  \BF^{\rm H}$ is the Moore-Penrose pseudo-inverse\footnote{the columns of $\BF$ are linearly independent.} (see \cite{Book:Meyer}, Chapter $5$, pp. $422$). Therefore, $\BD$ is an idempotent matrix with eigenvalues equal to either $0$ or $1$ \cite{Book:Horn} and hence, $[\BI - \BD]$ is also a positive semi-definite idempotent matrix. Note also that the matrix $\BE$ in (\ref{eq:matrixGLS3}) can be written as
\begin{eqnarray}
\nonumber \BE &=&\BA_{(i)} \mathrm{diag}(\hat{\calX}_{(i)}^{\rm H})[\BI-\BD]\mathrm{diag}(\hat{\calX}_{(i)}) \BA_{(i)}^{\rm H} \\
&=& \BB^{\rm H} [\BI-\BD] \BB
\end{eqnarray}

and
\begin{eqnarray}
\bf{z}^{\rm H} \BE \bf{z} = \bf{z}^{\rm H} \BB^{\rm H} [\BI-\BD] \BB \bf{z} =  (\BB \bf{z})^{\rm H} [\BI-\BD] (\BB \bf{z}) \geq 0
\end{eqnarray}
and so $\BE$ is Hermitian and positive semi-definite.

Let $\BU= [{\bf{u}}_1 \;\;{\bf{u}}_2\;\;\cdots\;\;{\bf{u}}_{L+1}]$ be a $(L+1) \times (L+1)$ unitary matrix where ${\bf{u}}_i$ is the $i$th eigenvector. then, $\BE = \BU \La \BU^{\rm H}$ where $\La$ is a diagonal matrix containing ordered eigenvalues of $\BE$ such that $\lambda_1\geq\lambda_2\geq\cdots\geq\lambda_{L+1}$. Let ${\bf{z}} = \BU^{\rm H}{\bf{v}}$, then the maximum eigenvalue of $\BE$ is given as
\begin{eqnarray}
\max_{||{\bf{v}}||_2=1} {\bf{v}}^{\rm H}\BE{\bf{v}} &=& \max_{||{\bf{z}}||_2=1} {\bf{z}}^{\rm H}\La{\bf{z}}\\
 &=& \max_{||{\bf{z}}||_2=1} \sum_{i=1}^{L+1}\lambda_i|{{z}}_i|^2\\
 &\leq& \max_{||{\bf{z}}||_2=1} \lambda_1\sum_{i=1}^{L+1}|{{z}}_i|^2 \\
 &\leq& \lambda_1 = \lambda_{\rm max}
\end{eqnarray}
The equality is attained when ${\bf{v}}$ is the eigenvector of $\lambda_{\rm max}$.
\end{proof}

\begin{lemma}
Given that $\BE$ has a positive maximum eigenvalue $\lambda_{\rm max}$  with corresponding unit-norm vector ${\bf{v}}$ of size $(L+1) \times 1$, then the maximum eigenvalue of $\BG_{(i)}$ in (\ref{eq:matrixGLS2}) is lower bounded by ${\bf{w}}^{\rm H} \BG_{(i)} {\bf{w}} = \rho\;\lambda_{\rm max}$ where
\begin{eqnarray}
{\bf{w}} = \begin{bmatrix}
  {\bf{v}}_{(L+1) \times 1} \\ {\bf{0}}_{i\times 1}
\end{bmatrix}
\end{eqnarray}
\end{lemma}
\begin{proof}
From Lemma $1$, the largest eigenvalue of $\BE$ is $\lambda_{\rm max}$. It follows that the largest eigenvalue of $\rho\BE$ is $\rho\lambda_{\rm max}$. Let $\lambda_{\rm max}'$ be the largest eigenvalue of $\BG_{(i)}$. From (\ref{eq:matrixGLS3}), we can see that $\rho\BE$ is a principal sub-matrix of $\BG_{(i)}$ (see \cite{Book:Meyer}, Chapter $7$, pp. $494$) and thus
\begin{equation}
\lambda_{\rm max}' \geq \rho\lambda_{\rm max}
\end{equation}
i.e., the largest eigenvalue of the principal sub-matrix $\rho\BE$ is smaller than or equal to the largest eigenvalue of $\BG_{(i)}$ (see \cite{Book:Meyer}, Chapter $7$, pp. $551$-$552$). Thus $\rho\lambda_{\rm max}$ is a lower bound on the largest eigenvalue of $\BG_{(i)}$.
\end{proof}

Note that $\BSigma_i$ is positive definite as it is a covariance matrix, hence it will have positive eigenvalues. From Lemma $2$, the maximum eigenvalue of $\BG_{(i)},\;\; \lambda_{\rm max}' \rightarrow \infty$ as $\rho \rightarrow \infty$. Thus the denominator in (\ref{eq:P_ISimplfiedLS}) grows to infinity in the limit $\rho \rightarrow \infty$ and
\begin{equation}
\lim_{\rho \rightarrow \infty} P_i \rightarrow 0 \label{eq:LSP_itendsTo0}
\end{equation}
From (\ref{Cii}) and (\ref{eq:LSP_itendsTo0}), we have
\begin{eqnarray}
\lim_{\rho \rightarrow \infty}C_i &\leq& 1+(|\Omega|^i - 1)\lim_{\rho \rightarrow \infty} P_i \\
\lim_{\rho \rightarrow \infty}C_i &\leq& 1
\label{CiiFinal}
\end{eqnarray}

\subsection{MAP cost}
The cost associated with the MAP solution of an erroneous sequence of symbols $\bar{\calX}_{(i)} \neq \hat{\calX}_{(i)}$ is given as (see \cite{Book:Sayed2003}, Chapter $11$, pp. $672$)
\begin{equation}
{M}_{\bar{\mathcal{X}}_{(i)}} = \calY_{(i)}^{\rm H}\Big(\mbox{{\emph{\textbf{I}}}} + \rho\;{\rm diag}(\bar{\calX}_{(i)}) \BA_{(i)}^{\rm H} \mbox{{\emph{\textbf{R}}}}_h \BA_{(i)}{\rm diag}(\bar{\calX}_{(i)}^{\rm H})\Big)^{-1} \calY_{(i)}
\end{equation}
Mathematically,
\begin{eqnarray}
\nonumber P_i &=&\mbox{Pr} ({M}_{\bar{\mathcal{X}}_{(i)}} \leq r) \\ \nonumber
P_i &=& \\ && \nonumber \hspace{-4em}\mbox{Pr} \bigg( \calY_{(i)}^{\rm H}\Big(\mbox{{\emph{\textbf{I}}}} + \rho\;{\rm diag}(\bar{\calX}_{(i)}) \BA_{(i)}^{\rm H} \mbox{{\emph{\textbf{R}}}}_h \BA_{(i)}{\rm diag}(\bar{\calX}_{(i)}^{\rm H})\Big)^{-1} \calY_{(i)} \leq r\bigg).\\ \label{P_i_inter1}
\end{eqnarray}
By matrix inversion lemma
\begin{eqnarray}
\nonumber && \left(\mbox{{\emph{\textbf{I}}}} + \sqrt{\rho}\;{\rm diag}(\bar{\calX}_{(i)}) \BA_{(i)}^{\rm H} \mbox{{\emph{\textbf{R}}}}_h \BA_{(i)}{\rm diag}(\bar{\calX}_{(i)}^{\rm H})\right)^{-1} \\ \nonumber
&=& \mbox{{\emph{\textbf{I}}}} - \rho\; {\rm diag}(\bar{\calX}_{(i)}) \BA_{(i)}^{\rm H} \Big[ \mbox{{\emph{\textbf{R}}}}_h^{-1} + \\ &&   \rho\;\BA_{(i)}{\rm diag}(\bar{\calX}_{(i)}^{\rm H}){\rm diag}(\bar{\calX}_{(i)}) \BA_{(i)}^{\rm H} \Big]^{-1} \BA_{(i)}{\rm diag}(\bar{\calX}_{(i)}^{\rm H})   \\
\nonumber &=& \mbox{{\emph{\textbf{I}}}} - {\rm diag}(\bar{\calX}_{(i)}) \BA_{(i)}^{\rm H} \Big[ \frac{1}{\rho}\; \mbox{{\emph{\textbf{R}}}}_h^{-1} + \\ && \nonumber \BA_{(i)}{\rm diag}(\bar{\calX}_{(i)}^{\rm H}){\rm diag}(\bar{\calX}_{(i)}) \BA_{(i)}^{\rm H} \Big]^{-1} \BA_{(i)}{\rm diag}(\bar{\calX}_{(i)}^{\rm H})  \\
&=& \mbox{{\emph{\textbf{I}}}} - \BD \label{P_i_inter2}
\end{eqnarray}
where
\begin{eqnarray}
\nonumber \BD &=& \\ \nonumber && \hspace{-3.75em}{\rm diag}(\bar{\calX}_{(i)}) \BA_{(i)}^{\rm H} \Big[ \frac{1}{\rho}\; \mbox{{\emph{\textbf{R}}}}_h^{-1} +  \BA_{(i)}{\rm diag}(\bar{\calX}_{(i)}^{\rm H}){\rm diag}(\bar{\calX}_{(i)}) \BA_{(i)}^{\rm H} \Big]^{-1}\\ && \BA_{(i)}{\rm diag}(\bar{\calX}_{(i)}^{\rm H})\label{eq:DofMAP}
\end{eqnarray}
Thus (\ref{P_i_inter1}) can be written as
\begin{eqnarray}
P_i = \mbox{Pr} \bigg( \calY_{(i)}^{\rm H}\Big( \mbox{{\emph{\textbf{I}}}} - \BD \Big) \calY_{(i)} \leq r\bigg)\label{P_i_inter3}
\end{eqnarray}
note that (\ref{P_i_inter3}) is of the same form as (\ref{P_i_inter1LS}). The only difference in the LS and MAP costs is the presence of the term $\frac{1}{\rho}\; \mbox{{\emph{\textbf{R}}}}_h^{-1}$ in (\ref{eq:DofMAP}). Also note that this term depends on the inverse of the SNR. For low SNR, the inverse term in (\ref{eq:DofMAP}) is always invertible due to the regularization term. At high SNR, the effect of regularization fades and inverse term in (\ref{eq:DofMAP}) is invertible. At high SNR, i.e., $\rho\rightarrow \infty$, $\frac{1}{\rho}\; \mbox{{\emph{\textbf{R}}}}_h^{-1} \rightarrow 0$ and $\BD$ of (\ref{P_i_inter2}) takes the same form as that of LS cost leading to (\ref{CiiFinal}).


Table \ref{tablecomplexityOFDM} lists the estimated computational
cost for our blind algorithm in the high SNR regime. Since there
is no backtracking, the total number of iterations is $N$, which
explains our calculations in Table \ref{tablecomplexityOFDM}. It
thus follows that the total number of operations needed for our
algorithm is of the order $O(LN)$ in high SNR regime. The pilot
based approach for channel estimation needs to invert an
$(L+1)\times(L+1)$ matrix $($assuming we need $L+1$ pilots to
estimate a channel of length $L+1)$ with a complexity of the order
$O(L^2)$. Since the cyclic prefix is a fixed fraction of the OFDM
symbol $(L=N/m$ with $m$ typically set to $m=4$ or $8)$ we see
that the complexity of the two approaches become comparable in the
high SNR regime.

\begin{table}
  \centering
  \caption{Total computational cost of the ML blind and training based algorithms at high SNR}\label{tablecomplexityOFDM}
\begin{tabular}{|c|c|c|}
  \hline
  \textbf{Algorithm} & $\times$ & $+$\\
  \hline
  \hline
    & & \\
  Blind Algorithm & $(3L^2+11L+17)N$ & $(2L^2+5L+4)N$ \\
    & & \\
  \hline
    Blind algorithm & & \\
   with  & $(4L+13)N$ & $(2L+4)N$ \\
    carrier reordering & & \\
  \hline
   Training based & & \\
   algorithm \cite{VCT:Naffouri2002} & $4L^2+17L+13$ & $2L^2+6L+4$ \\
    & & \\
  \hline
\end{tabular}
\end{table}

\section{Simulation Results} \label{Sim2}

We consider an OFDM system with $N=16$, or $64$ subcarriers and a CP of length $L = \frac{N}{4}$. The uncoded data symbols are modulated using BPSK, $4$-QAM, or $16$-QAM. The constructed OFDM signal then passes through a channel of length $L+1$, which is assumed to be block fading (i.e., constant over one OFDM symbol but fades independently from one symbol to another) and whose taps follow an exponential decay profile ($E[|h(t)|^2]=e^{-0.2t}$).

\subsection{Bench marking}

We compare the performance of our algorithm against the following receivers
\begin{enumerate}
  \item the subspace-based\footnote{The block fading assumption is maintained for all simulations. However, for the subspace blind receiver of \cite{IEEETSP:Muquet2002} to work, the channel needs to stay constant over a sequence of OFDM symbols. For this particular receiver, the channel was kept fixed over $50$ OFDM symbols.} blind receiver of \cite{IEEETSP:Muquet2002},
  \item the sphere decoding based receiver of \cite{IEEETC:Cui2006},
  \item a receiver that acquires the channel through training with $L+1$ pilots and a priori channel correlation $\BR_h$ \cite{VCT:Naffouri2002},
  \item the ML receiver that acquires data through exhaustive search.
\end{enumerate}
The simulations are averaged over $500$ Monte-Carlo runs.

Figure \ref{Fig:BPSK_all_saq} compares the BER performance of our algorithm with the aforementioned algorithms for an OFDM system with $N=16$ subcarriers and BPSK data symbols. Note in particular that our blind algorithm outperforms both the subspace and sphere decoding algorithms and almost matches the performance of the exhaustive search algorithm for low and high SNR, which confirms the ML nature of the algorithm.

Figure \ref{Fig:4QAM_all_saq}, which considers the $4$-QAM case, shows the same trends observed for the BPSK case of Figure \ref{Fig:BPSK_all_saq}.

\begin{figure}[htb]
\begin{center}
\epsfxsize = 3.5 true in \epsfbox{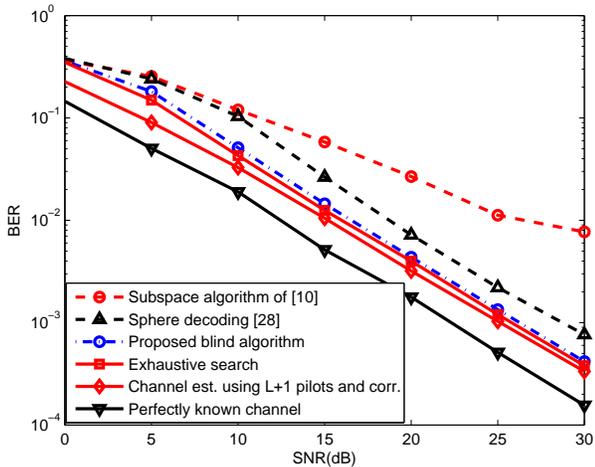}
\caption{\small BER vs SNR for BPSK OFDM over a Rayleigh channel
with $N=16$ and $L=3$} \label{Fig:BPSK_all_saq}
\end{center} \end{figure}

\begin{figure}[htb]
\begin{center}
\epsfxsize = 3.5 true in \epsfbox{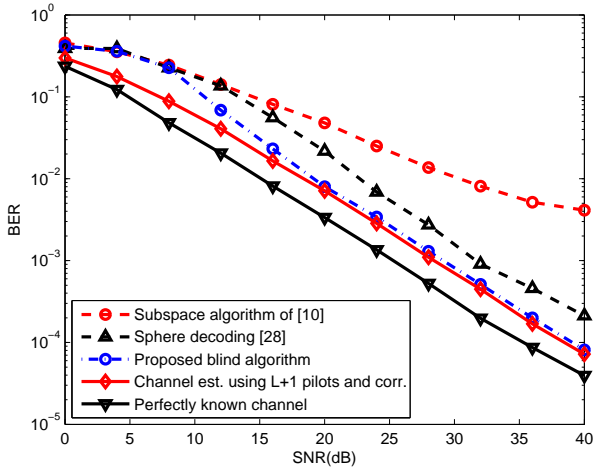}
\caption{\small BER vs SNR for $4$-QAM OFDM over a Rayleigh channel
with $N=16$ and $L=3$} \label{Fig:4QAM_all_saq}
\end{center} \end{figure}

Figure \ref{4QAM} considers a more realistic OFDM symbol length ($N=64$), drawn from a $4$-QAM constellation and allows the SNR to grow to $45$ dB. Our blind algorithm shows no error floor signs, which is characteristic of non-ML methods. Furthermore, the algorithm beats the training-based method and follows closely the performance of the perfect channel case. Figure \ref{16QAM64} shows the results with $N=64$ subcarriers and $16$-QAM data symbols for SNR as large as $50$ dB. Again,  the proposed blind algorithm does not reach an error floor.

\begin{figure}[htb]
\begin{center}
\epsfxsize = 3.5 true in \epsfbox{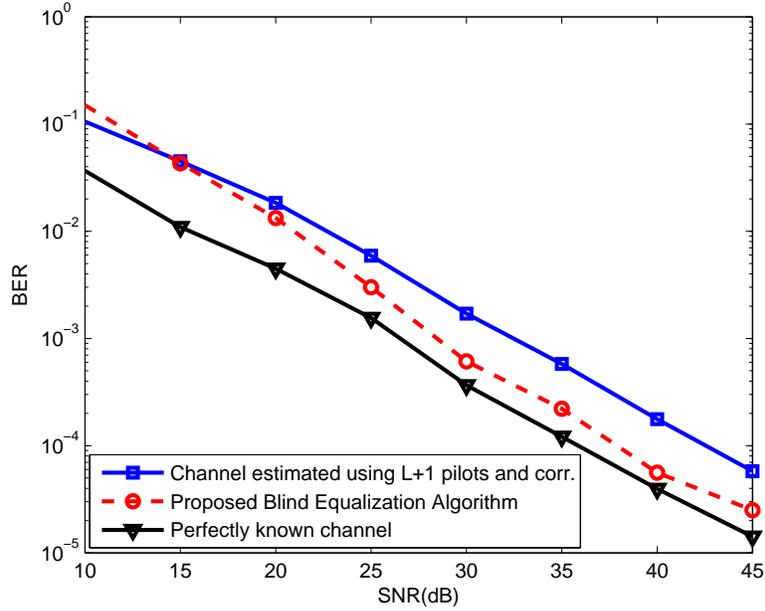}
\caption{\small BER vs SNR for $4$-QAM OFDM over a Rayleigh channel
with $N=64$ and $L=15$} \label{4QAM}
\end{center} \end{figure}

\begin{figure}[htb]
\begin{center}
\epsfxsize = 3.5 true in \epsfbox{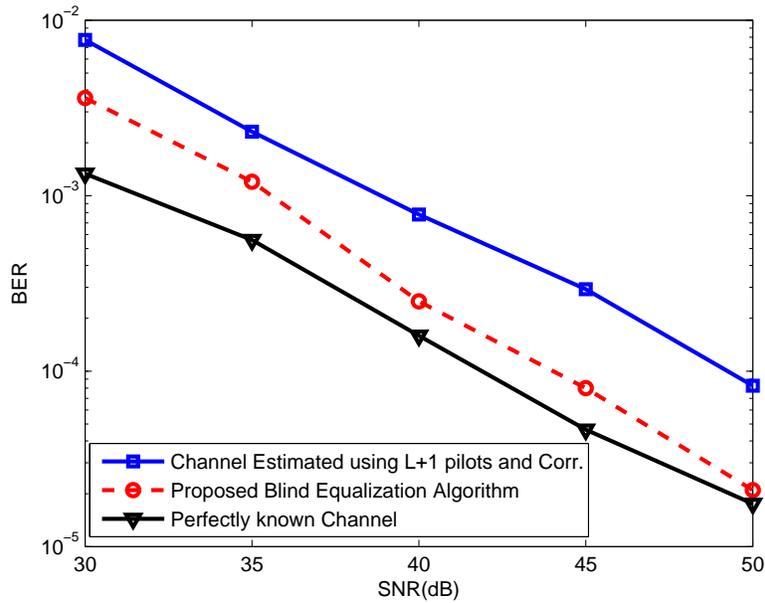}
\caption{\small BER vs SNR for $16$-QAM OFDM over a Rayleigh
channel with $N=64$ and $L=15$} \label{16QAM64}
\end{center} \end{figure}

\subsection{Low-Complexity Variations}

In this subsection, we investigate the low-complexity variants of our algorithm. Specifically, we consider the performance of the blind algorithm with
\begin{enumerate}
  \item ${\BP}_i$ set to $\II$,
  \item ${\BP}_i$ set to $\II$ with subcarrier reordering
\end{enumerate}
Figure \ref{Fig:BPSK_vari_saq} exhibits the comparisons for the various algorithms for BPSK and $N=16$. Note that with ${\BP}_i$ set to $\II$ arbitrarily, the performance of the blind algorithm deteriorates and the BER reaches an error floor. Contrast this with the algorithm variant that uses subcarrier reordering as well, and note that the performance of this variant follows closely the performance of the exact blind algorithm. Also note that the BER of both of these algorithms beats that of the sphere decoding algorithm of \cite{IEEETC:Cui2006}. The same trends are observed in Figure \ref{Fig:4QAM_vari_saq}, which considers the $4$-QAM case.

\begin{figure}[htb]
\begin{center}
\epsfxsize = 3.5 true in \epsfbox{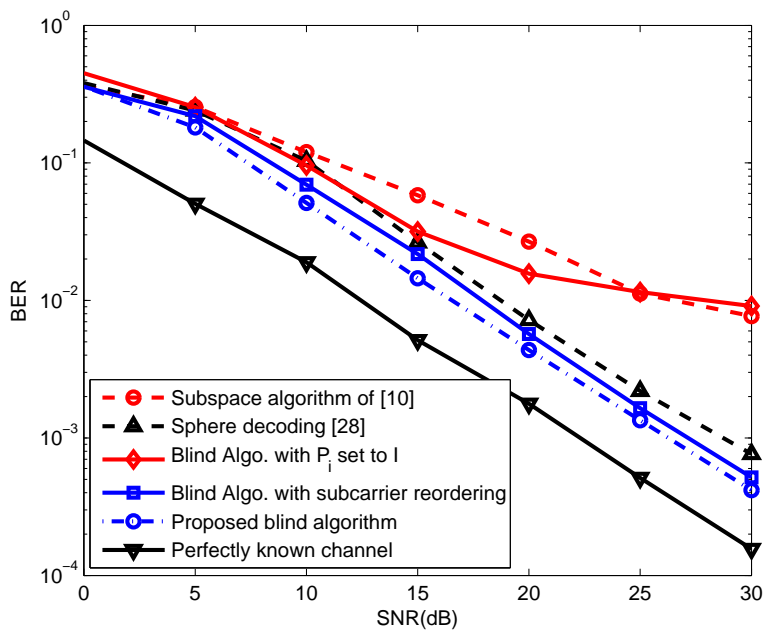}
\caption{\small Comparison of low-complexity algorithms for
BPSK OFDM with $N=16$ and $L=3$} \label{Fig:BPSK_vari_saq}
\end{center} \end{figure}

\begin{figure}[htb]
\begin{center}
\epsfxsize = 3.5 true in \epsfbox{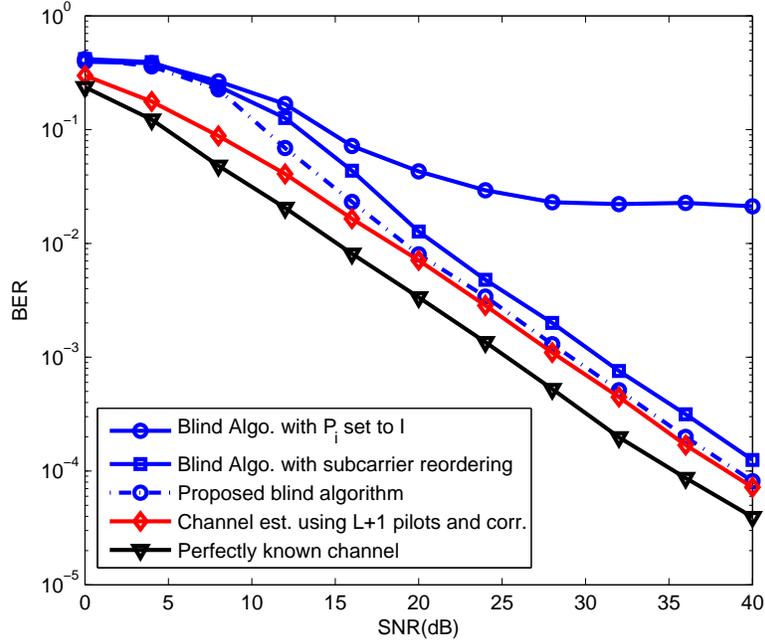}
\caption{\small Comparison of low-complexity algorithms for
$4$-QAM OFDM with $N=16$ and $L=3$} \label{Fig:4QAM_vari_saq}
\end{center} \end{figure}

Figure \ref{Fig:time_all_methods} compares the average runtime of various algorithms as a function of the SNR. Note first that the extreme cases are the training-based receiver and the exhaustive search receiver, both of which are independent of the SNR. The runtime of the proposed algorithm decreases with the SNR and is sandwiched in-between the run time of the sphere decoding algorithm and that of the subspace algorithm for all values of the SNR\footnote{The runtime of the subspace algorithm is adjusted to account for the fact that it requires the channel to be constant over a block of $L+1$ OFDM symbols.}. Note that in the high SNR regime our algorithm runs at the same speed as the subspace algorithm.

Figure \ref{Fig:time_all_modu} shows the average runtime of the proposed algorithm with $N=16$ for various modulation schemes (BPSK, $4$-QAM and $16$-QAM). It is clear from the figure that the average runtime decreases considerably at higher SNR values.

\begin{figure}[htb]
\begin{center}
\epsfxsize = 3.5 true in \epsfbox{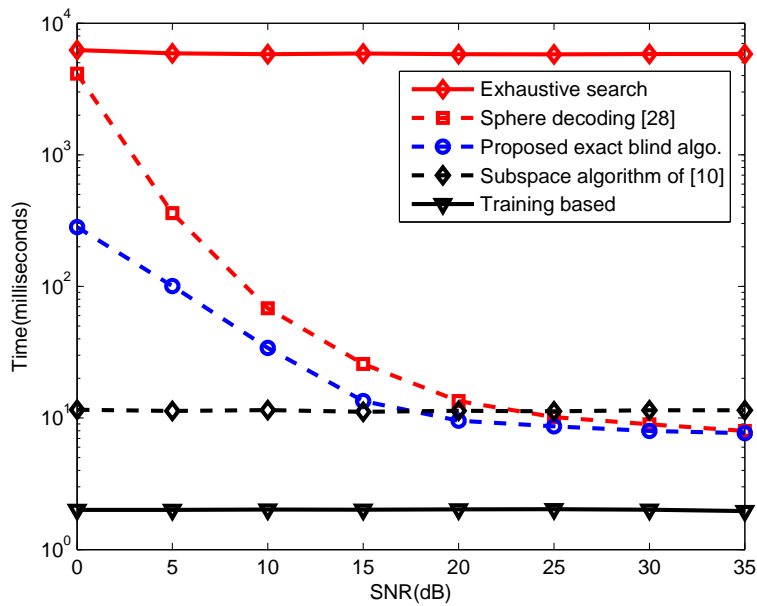}
\caption{\small Average time comparison for BPSK data symbols with $N=16$
and $L=3$} \label{Fig:time_all_methods}
\end{center}
\end{figure}
\begin{figure}[htb]
\begin{center}
\epsfxsize = 3.5 true in \epsfbox{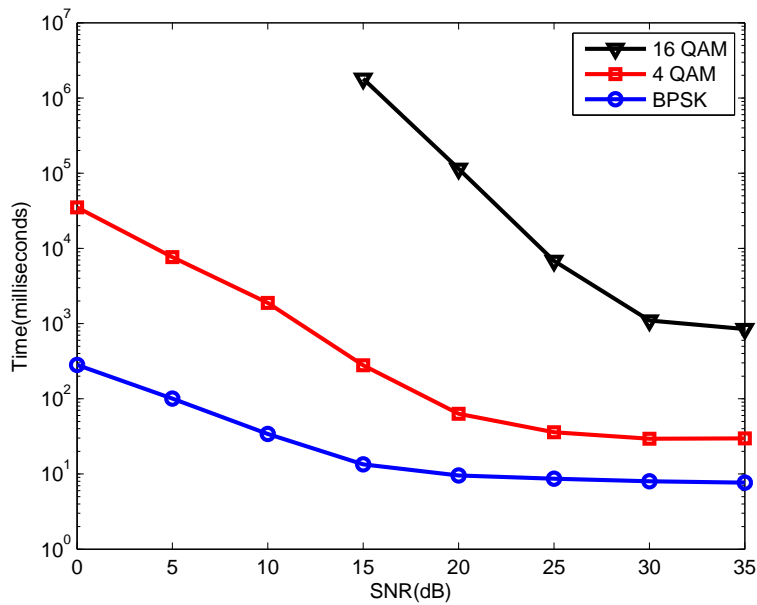}
\caption{\small Average Time Comparison for our Blind Algorithm
for Different Modulation with $N=16$ and $L=3$} \label{Fig:time_all_modu}
\end{center}
\end{figure}

\section{Conclusion}\label{con2}

In this paper, we have  proposed a low-complexity blind algorithm that is able to deal with channels that change on a symbol by symbol basis allowing it to deal with fast block fading channels. The algorithm works for general constellations and is able to recover the data from output observations only. Simulation results demonstrate the favorable performance of the algorithm for general constellations and show that its performance matches the performance of the exhaustive search for small values of $N$.

We have also proposed an approximate blind equalization method (avoiding $\BP_i$ with subcarrier reordering) to reduce the computational complexity. As evident from the simulation results, this approximate method performs quite close to the exact blind algorithm and can work properly without a priori knowledge of the channel statistics. Finally, we study the complexity of our blind algorithm and show that it becomes especially low in the high SNR regime.


\section*{Acknowledgment}
The author would like to acknowledge the support provided by the Deanship of Scientific Research (DSR) at King Fahd University of Petroleum \& Minerals (KFUPM) for funding this work through project No. FT$111004$.



\end{document}